\documentclass[a4paper]{article}

\usepackage[utf8]{inputenc}
\usepackage[left=25mm,right=25mm,top=30mm,bottom=30mm]{geometry}

\usepackage{xparse}
\usepackage{amsmath}
\usepackage{amssymb}
\usepackage{amsthm}
\usepackage{bbm}
\usepackage{authblk}

\usepackage[numbers]{natbib}

\usepackage{todonotes}
\usepackage{pi}
\usetikzlibrary{shapes.geometric}
\usetikzlibrary{decorations.pathmorphing}
\usetikzlibrary{decorations.shapes}

\usepackage[pdfpagelabels]{hyperref}
\hypersetup{colorlinks,citecolor=black,linkcolor=black,urlcolor=black}

\newtheorem{theorem}{Theorem}
\newtheorem{corollary}[theorem]{Corollary}
\newtheorem{lemma}[theorem]{Lemma}

\DeclareDocumentCommand\R{}{\mathbb{R}}
\DeclareDocumentCommand\Z{}{\mathbb{Z}}
\DeclareDocumentCommand\N{}{\mathbb{N}}
\DeclareDocumentCommand\bincoeff{mm}{\genfrac{(}{)}{0pt}{}{#1}{#2}}
\DeclareDocumentCommand\setdef{mo}{\left\{#1\IfNoValueTF{#2}{}{ \mid #2}\right\}}
\DeclareDocumentCommand\conv{o}{\operatorname{conv}\IfValueTF{#1}{\left(#1\right)}{}}
\DeclareDocumentCommand\onevec{o}{\IfNoValueTF{#1}{\mathbbm{1}}{\mathbbm{1}_{#1}}}
\DeclareDocumentCommand\orderO{m}{\mathcal{O}\left(#1\right)}

\title{Parity Polytopes and Binarization}
\author[1]{Dominik Ermel}
\affil[1]{Otto-von-Guericke-Universität Magdeburg, dominik.ermel@st.ovgu.de}
\author[2]{Matthias Walter}
\affil[2]{RWTH Aachen University, walter@or.rwth-aachen.de}

\begin{document}

\maketitle

\begin{abstract}
  We consider generalizations of parity polytopes whose
  variables, in addition to a parity constraint,
  satisfy certain ordering constraints.
  More precisely, the variable domain is partitioned into $k$
  contiguous groups, and within each group, we require $x_i \geq x_{i+1}$
  for all relevant $i$.
  Such constraints are used to break symmetry
  after replacing an integer variable by a sum of binary variables, so-called binarization.
  We provide extended formulations for such polytopes,
  derive a complete outer description, and present a separation algorithm
  for the new constraints.
  It turns out that applying binarization and only enforcing parity
  constraints on the new variables is often a bad idea.
  For our application, an integer programming model for the graphic traveling salesman problem,
  we observe that parity constraints do not improve the dual bounds,
  and we provide a theoretical explanation of this effect.
\end{abstract}

\pagebreak[3]
\section{Introduction}
\label{SectionIntroduction}

\DeclareDocumentCommand\Pord{m}{P_{\text{ord}}^{#1}}
\DeclareDocumentCommand\Xord{m}{X_{\text{ord}}^{#1}}

The term \emph{binarization} refers to techniques to reformulate mixed-integer linear programs by replacing general integer (bounded) variables
by sets of binary variables.
\citeauthor{Roy07} compared several approaches with the goal to separate strong cutting planes in the reformulation and project it back~\cite{Roy07}.
Recently,~\citeauthor{BonamiM14} continued this line of work using simple split disjunctions of ranks~1 and~2 to generate cutting planes~\cite{BonamiM14}.
The strongest bounds were obtained by a binarization in which an integer variable $z \in \{0,1,\dotsc,n\}$ is replaced by the sum of $n$
binary variables $x_1, \dotsc, x_n$, and symmetry-breaking constraints
\begin{gather}
  \label{EquationBinarizationSymmetry}
  1 \geq x_1 \geq \dotsc \geq x_n \geq 0
\end{gather}
are added.
We denote by $\Xord{n}$ the set of ordered binary vectors $x \in \{0,1\}^n$ satisfying 
\eqref{EquationBinarizationSymmetry} and by $\Pord{n}$ their convex hull.
Note that Inequalities~\eqref{EquationBinarizationSymmetry} already describe $\Pord{n}$
since the corresponding matrix is totally unimodular
(every row has at most one $+1$ and at most one $-1$, see Theorem~19.3 in~\cite{Schrijver86})
and the right-hand side is integral.
We call the vectors in $\Xord{n}$ \emph{ordered binary vectors}.

\DeclareDocumentCommand\Peven{m}{P_{\text{even}}^{#1}}

In this article we consider the strong binarization in combination with parity constraints on the affected integer variables.
For this we define the \emph{ordered even parity polytope} for a vector $r \in \N^k$ as
\begin{gather*}
  \Peven{r} := \conv\{ (x^{(1)}, \dotsc, x^{(k)}) \in \Xord{r_1} \times \dotsb \times \Xord{r_k} \mid \sum_{i=1}^k \sum_{j=1}^{r_i} x^{(i)}_j \text{ even} \} .
\end{gather*}
The integer points in $\Peven{r}$ are precisely those binary vectors of length $n := r_1 + \dotsb + r_k$
whose entry-wise sum is even and which are ordered within each of the groups defined by $r$.
In the special case of $r = \onevec[n]$ (the all-ones vector), $\Peven{r}$ is the
even parity polytope, which has the following outer description~\cite{Jeroslow75} (with $[n] := \setdef{1,2,\dotsc,n}$):
\begin{gather}
  \label{InequalityParity}
  \Peven{\onevec[n]} = \{ x \in [0,1]^n \mid \sum_{i \in [n] \setminus F} x_i + \sum_{i \in F} (1-x_i) \geq 1 \text{ for all $F \subseteq [n]$ with $|F|$ odd} \}
\end{gather}

\paragraph{Outline.}
We start by describing an extended formulation for $\Peven{r}$ for arbitrary $r \in \N^k$,
which implies that the associated optimization problem can be solved in polynomial time (in the dimension of $\Peven{r}$).
In Section~\ref{SectionOuterDescription} we state and prove the main result of this paper, a complete outer description for $\Peven{r}$ in the original space.
After transferring the result to odd parities in Section~\ref{SectionOddParities}
we derive in Section~\ref{SectionSeparation} a simple linear-time algorithm that solves the respective separation problems.
Section~\ref{SectionGTSP} is dedicated to an application in which binarization is applied
to an integer programming model for the graphic traveling salesman problem.
During computations for this model we observed that the corresponding parity constraints
did not improve the dual relaxation value.
In Section~\ref{SectionCheating} we provide a theoretical explanation of this effect.

\pagebreak[3]
\section{Extended Formulations}
\label{SectionExtendedFormulation}

\DeclareDocumentCommand\DoutArcs{}{\delta^{\text{out}}}
\DeclareDocumentCommand\DinArcs{}{\delta^{\text{in}}}
\DeclareDocumentCommand\Pflow{mmm}{P_{\text{flow}}^{{#2}\text{-}{#3}}(#1)}

In this section we review well-known extended formulations of parity polytopes based on dynamic programming,
and adapt them to the ordered case.
The linear-size formulation for $\Peven{\onevec[n]}$ introduced by~\citeauthor{CarrK05}~\cite{CarrK05} works as follows.
We define a digraph $D = (V,A)$ with nodes $V := \setdef{0,1,2,\dotsc,n} \times \setdef{0,1} \setminus \setdef{ (0,1), (n,1) }$
and arcs $A := \setdef{ ((i-1,\alpha),(i,\beta)) \in V \times V }[ i \in [n], \alpha,\beta \in \setdef{0,1} ]$.
We call the nodes $s := (0,0)$ the source and $t := (n,0)$ the sink.
Clearly, the set $\Pflow{D}{s}{t}$ of $s$-$t$-flows $y \in \R_+^A$ of flow-value $1$ in any such acyclic digraph is described by
the linear constraints
\begin{alignat*}{10}
  y(\DoutArcs(v)) - y(\DinArcs(v))      &= 0    \quad\text{ for all } v \in V \setminus \setdef{s,t} \\
  y(\DoutArcs(s))                       &= 1    \\
  y(\DinArcs(t))                        &= 1    \\
  y_a                                   &\geq 0 \quad\text{ for all } a \in A,
\end{alignat*}

where we denote by $\DoutArcs(\cdot)$ resp.\ $\DinArcs(\cdot)$ the set of outgoing resp.\ incoming arcs.
Furthermore, again due to the total unimodularity of $D$'s incidence matrix (see Application~19.2 in~\cite{Schrijver86}),
$\Pflow{D}{s}{t}$ is an integral polytope.
We claim that the following map projects $\Pflow{D}{s}{t}$ onto $\Peven{\onevec[n]}$:
\begin{gather*}
  \pi : \R^A \to \R^n \qquad\text{ with }\qquad \pi(y)_i :=
  \begin{cases}
    y_{((0,0),(1,1))} & \text{ for } i = 1 \\
    y_{((i-1,0),(i,1))} + y_{((i-1,1),(i-1,0))} & \text{ for } 2 \leq i \leq n-1 \\
    y_{((n-1,1),(n,0))} & \text{ for } i = n
  \end{cases}
\end{gather*}
It is not hard to see that the integer points in $\pi(\Pflow{D}{s}{t})$ are precisely the binary vectors with an even number of $1$'s,
since an $s$-$t$-path in $D$ must use an even number of diagonal arcs.
The claim follows because $\Pflow{D}{s}{t}$ is integral and because $\pi$ maps integral points to integral points.
Figure~\ref{FigureParityFlowExtension} illustrates the constructed network and the projection (denoted by $r_i = 1$).

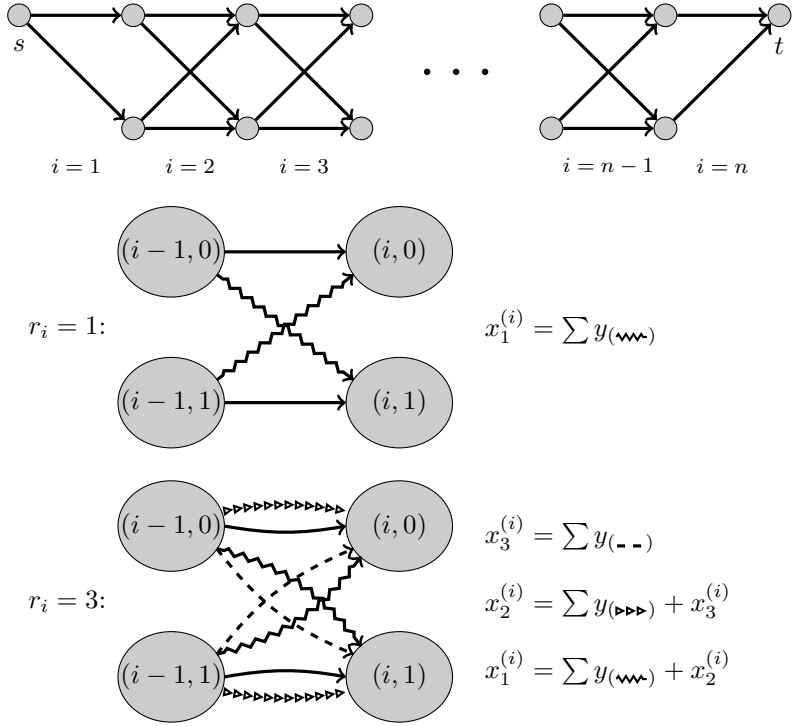
\begin{figure}[ht]
  \begin{center}
    \begin{tikzpicture}
      \piInput{tikz-extended-formulations-ordered-parity}[job=network]
    \end{tikzpicture}

    \vspace{1mm}

    \begin{tikzpicture}
      \piInput{tikz-extended-formulations-ordered-parity}[job=default-level]
    \end{tikzpicture}

    \vspace{1mm}

    \begin{tikzpicture}
      \piInput{tikz-extended-formulations-ordered-parity}[job=ordered-level]
    \end{tikzpicture}
  \end{center}
  \caption{Flow-Extensions of (Ordered) Even Parity Polytopes.}
  \label{FigureParityFlowExtension}
\end{figure}

We can now modify the network such that a similar projection yields $\Peven{r}$ for arbitrary vectors $r \in \N^k$.
The new construction works as follows.
The node set of the new digraph is the same as the old digraph for $n = k$, i.e.,
the node set of a $(2 \times (k+1))$-grid graph without the lower-left and lower-right corner nodes.
Again, $s$ and $t$ are the left-most and right-most nodes, respectively.
The arcs again go from left to right, and the arcs from the $k$ grid-layers correspond to the
sets of variables $x^{(i)}$ for $i \in [k]$ in the sense that the projection of an arc variable
only has impact on a subset of the corresponding $x$-variables.
For every layer $i \in [k]$ and every node $(i-1,\alpha)$ (for $\alpha=0$ if $i=1$ or $i=k$, and for $\alpha \in \setdef{0,1}$ otherwise)
on the left of this layer there are $r_i+1$ outgoing arcs, each associated to one of the possible
values $\ell \in \setdef{0,1,2,\dotsc,r_i}$. The arcs with even $\ell$ all go to the same node $(i,\alpha)$,
and the arcs with odd $\ell$ all go to node $(i,1-\alpha)$.

The projection map is defined by a 0/1-matrix. We describe it column-wise, i.e.,
we state for each arc in which rows (corresponding to $x$-variables) it contains a $1$-entry.
Consider an arc going from node $(i-1,\alpha)$ to node $(i,\beta)$ that is associated to $\ell \in \setdef{0,1,2,\dotsc,r_i}$,
and thus satisfies $\alpha + \beta \equiv \ell \pmod 2$.
The column of the projection matrix that corresponds to this arc contains $1$'s in rows corresponding to the variables $x^{(i)}_{1}, x^{(i)}_{2}, \dotsc, x^{(i)}_{\ell}$.
The last part in Figure~\ref{FigureParityFlowExtension} illustrates a layer for $r_i = 3$.

The correctness of the modification is justified similarly to the ordinary construction.
First, the flow polytope is integral due to total unimodularity of its constraint matrix.
Second, the projection matrix is integral.
Third, the integral $s$-$t$-flows map to the correct vectors $x$.
To see this, observe that for every $i \in [k]$, the corresponding arcs of the network
correspond to the vectors in $\Xord{r_i}$,
and that the number of diagonal arcs in such a path must be an even number.

\pagebreak[3]
\section{Outer Description}
\label{SectionOuterDescription}

We begin with the simple observation that for ordered binary vectors,
parity can be measured using a linear function.
For $n \in \N$ we define $f : \Pord{n} \to \R$
via $f(x) := \sum_{i=1}^n (-1)^{i-1} x_i$.
From
\begin{align*}
  f(x) &= \underbrace{x_1 - x_2}_{\geq~0} + \underbrace{x_3 - x_4}_{\geq~0} + x_5 - \dotsb \pm x_n \qquad\text{and} \\
  f(x) &= \underbrace{x_1}_{\leq~1} - \underbrace{(x_2 - x_3)}_{\geq~0} - \underbrace{(x_4 - x_5)}_{\geq~0} + \dotsb \pm x_n
\end{align*}
we obtain that $f$ maps $\Pord{n}$ into $[0,1]$.
Thus, for $x \in \Xord{n}$, the value $f(x)$ is binary, and it is easy to check that $f(x) = 1$ holds if and only if $x$ has
an odd number of $1$'s.
Note that we will use $f$ for different values of $n$, and assume that it is clear from the context.
We can now state our main theorem.
\begin{theorem}
  \label{TheoremOrderedParityOuterDescriptionEven}
  Let $r \in \N^k$.
  The ordered even parity polytope $\Peven{r}$ is equal to the set of those
  points $(x^{(1)}, \dotsc, x^{(k)}) \in \Pord{r_1} \times \dotsb \times \Pord{r_k}$
  that satisfy the inequalities
  \begin{gather}
    \label{InequalityOrderedParity}
    \sum_{i \in [k] \setminus F} f(x^{(i)}) + \sum_{i \in F} (1-f(x^{(i)})) \geq 1
  \end{gather}
   for all $F \subseteq [k]$ with $|F|$ odd.
\end{theorem}

In order to prove the theorem we will make use of the following lemma.

\begin{lemma}
  \label{TheoremGlueingLemma}
  Let $X_0,X_1 \subseteq \R^m$ and $Y_0,Y_1 \subseteq \R^n$ be finite sets.
  Then the polytope
  \begin{gather*}
    \conv[(X_0 \times \setdef{0} \times Y_0) \cup (X_1 \times \setdef{1} \times Y_1)]
  \end{gather*}
  is equal to the polytope
  \begin{gather*}
    \conv[(X_0 \times \setdef{0}) \cup (X_1 \times \setdef{1})] \times \R^n
    ~~\cap~~
    \R^m \times \conv[(\setdef{0} \times Y_0) \cup (\setdef{1} \times Y1)] .
  \end{gather*}
\end{lemma}

\begin{proof}[Proof of Lemma~\ref{TheoremGlueingLemma}]
  The inclusion ``$\subseteq$'' is simple
  since for $i \in \setdef{0,1}$, every vector $(x,i,y) \in X_i \times \setdef{i} \times Y_i$
  satisfies $(x,i) \in X_i \times \setdef{i}$ and $(i,y) \in \setdef{i} \times Y_i$.

  To see ``$\supseteq$'', consider a point $(x,\lambda,y)$ from the second set,
  which clearly satisfies $0 \leq \lambda \leq 1$.
  Thus $(x,\lambda)$ is a convex combination of a point $(x^{(0)},0) \in \conv[X_0] \times \setdef{0}$ and a point
  $(x^{(1)},1) \in \conv[X_1] \times \setdef{1}$.
  From the $\lambda$-coordinate we obtain
  $x = \lambda x^{(0)} + (1-\lambda) x^{(1)}$.
  Similarly, there exist points $y^{(0)} \in \conv[Y_0]$ and $y^{(1)} \in \conv[Y_1]$ such that
  $y = \lambda y^{(0)} + (1-\lambda) y^{(1)}$ holds.
  This implies
  that $(x,\lambda,y) = \lambda (x^{(0)},0,y^{(0)}) + (1-\lambda) (x^{(1)},1,y^{(1)})$ holds,
  which concludes the proof.
\end{proof}

In our situation this lemma implies that if we ``glue together'' arbitrary 0/1-polytopes at a \emph{single} coordinate,
then the outer description of the resulting polytope is obtained by
taking the union of the outer descriptions of the two component polytopes.
We can now prove Theorem~\ref{TheoremOrderedParityOuterDescriptionEven}.

\begin{proof}[Proof of Theorem~\ref{TheoremOrderedParityOuterDescriptionEven}]
  We start with the validity of Inequalities~\eqref{InequalityOrderedParity} for $\Peven{r}$.
  Let $F \subseteq [k]$ be of odd cardinality.
  Since $f(x^{(i)}) \in \setdef{0,1}$ for every $x^{(i)} \in \Xord{r_i}$,
  the inequality can only be violated if all summands are equal to $0$, i.e.,
  if $f(x^{(i)}) = 0$ for all $i \in [k] \setminus F$
  and $f(x^{(i)}) = 1$ for all $i \in F$.
  This in turn means that $x^{(i)}$ contains an odd number of $1$'s
  if and only if $i \in F$ holds, contradicting the fact that
  the overall number of $1$'s is even because $|F|$ is odd.

  We now turn to the proof that Inequalities~\eqref{InequalityOrderedParity}
  cut off those points in $\Xord{r_1} \times \dotsb \times \Xord{r_k}$
  that have an odd number of $1$'s.
  Let $x = (x^{(1)}, \dotsc, x^{(k)})$ be such a point and define
  $F$ to be the set of indices $i \in [k]$
  for which $x^{(i)}$ has an odd number of $1$'s.
  Since $x$ has an odd number of $1$'s, we have that $|F|$ is odd.
  Furthermore, the left-hand side of \eqref{InequalityOrderedParity}
  is equal to $0$, i.e., $x$ violates it.

  It remains to prove that Inequalities~\eqref{InequalityOrderedParity}
  together with the inequalities that define each $\Pord{r_i}$
  for $i = 1,2,\dotsc,k$ yield an integral polytope.
  We show this by induction on the number $\ell$ of components of $r$
  that are strictly greater than $1$.
  For $\ell = 0$ we have $r = \onevec[n]$, and $\Peven{r}$
  is, according to \eqref{InequalityParity}, the even parity polytope.

  Let $\ell \in \N$.
  W.l.o.g., we can assume that $r_k \geq 2$ holds
  since otherwise we permute the variables suitably.
  By induction hypothesis we know that
  $$\Peven{(r_1,\dotsc,r_{k-1},1)} = \setdef{ (x^{(1)},\dotsc,x^{(k-1)},\lambda) \in \Pord{r_1} \times \dotsb \times \Pord{r_{k-1}} \times [0,1]
    }[ (x^{(1)},\dotsc,x^{(k-1)},\lambda) \text{ satisfies \eqref{InequalityOrderedParity}} ]$$
  holds, and is a $0/1$-polytope.
  Furthermore, the polytope
  $$Q := \setdef{ (\lambda,x^{(k)}) \in [0,1] \times \Pord{r_k} }[ \lambda = f(x^{(k)}) ]$$
  is affinely isomorphic to $\Pord{r_k}$.
  Since the latter is a $0/1$-polytope
  and since $f$ maps $\Xord{r_k}$ to $\setdef{0,1}$, $Q$ is also a $0/1$-polytope.
  Applying Lemma~\ref{TheoremGlueingLemma} to $\Peven{(r_1,\dotsc,r_{k-1},1)}$ and $Q$
  yields that
  \begin{multline*}
    R := \{ (x^{(1)},\dotsc,x^{(k-1)},\lambda,x^{(k)}) \in \Pord{r_1} \times \dotsb \times \Pord{r_{k-1}} \times [0,1] \times \Pord{r_k}
      \mid \lambda = f(x^{(k)}) \text{ and } \\
      (x^{(1)},\dotsc,x^{(k)}) \text{ satisfies \eqref{InequalityOrderedParity}} \}
  \end{multline*}
  is an integral polytope.
  The orthogonal projection of $R$ onto the $x$-variables is again an integral polytope
  and due to the equation $\lambda = f(x^{(k)})$ it is equal to the one in question.
\end{proof}

\pagebreak[3]
\section{Odd Parities}
\label{SectionOddParities}

\DeclareDocumentCommand\Podd{m}{P_{\text{odd}}^{#1}}

In this section we derive the outer description of the ordered \emph{odd} polytopes.
The latter are defined as
\begin{gather*}
  \Podd{r} := \conv\{ (x^{(1)}, \dotsc, x^{(k)}) \in \Xord{r_1} \times \dotsc \times \Xord{r_k} \mid \sum_{i=1}^k \sum_{j=1}^{r_i} x^{(i)}_j \text{ odd} \},
\end{gather*}
again for $r \in \N^k$.
The result corresponding to Theorem~\ref{TheoremOrderedParityOuterDescriptionEven} is the following:

\begin{corollary}
  \label{TheoremOrderedParityOuterDescriptionOdd}
  The ordered odd parity polytope for $r \in \N^k$ is the set of $(x^{(1)}, \dotsc, x^{(k)}) \in \Pord{r_1} \times \dotsc \times \Pord{r_k}$
  that satisfy Inequalities~\eqref{InequalityOrderedParity}
   for all $F \subseteq [k]$ with $|F|$ even.
\end{corollary}

In principle we could repeat the proof for Theorem~\ref{TheoremOrderedParityOuterDescriptionEven},
and exchange ``odd'' and ``even'' suitably.
Instead we present a proof that uses the previous result.

\begin{proof}
  We claim that
  $\Podd{(r_1,\dotsc,r_k)}$ is the projection of
  the face defined by $x^{(k+1)}_1 = 1$ of the polytope $\Peven{(r_1, \dotsc, r_k,1)}$
  onto the variables $x^{(1)},\dotsc,x^{(k)}$.
  This is justified by the fact that fixing the last component to $1$
  exchanges the meaning of ``odd'' and ``even'' for the sum of the (remaining) components.

  From Theorem~\ref{TheoremOrderedParityOuterDescriptionEven}
  we know the outer description of $\Peven{(r_1, \dotsc, r_k,1)}$,
  so it remains to carry out the projection.
  The only constraints that involve the last component are the bounds $0 \leq x^{(k+1)} \leq 1$
  and the Inequalities~\eqref{InequalityOrderedParity}
  for odd-cardinality sets $F \subseteq [k+1]$.
  The former need not be considered, and if $k+1 \notin F$ holds,
  then the latter inequalities are redundant since they read
  $\sum_{i \in [k] \setminus F} f(x^{(i)}) + 1 + \sum_{i \in F} (1-f(x^{(i)})) \geq 1$,
  and for $x^{(i)} \in \Pord{r_i}$ we know that $f(x^{(i)})$ and $1 - f(x^{(i)})$ are nonnegative.
  It remains to consider the case of $k+1 \in F$,
  in which the inequality is
  $\sum_{i \in [k] \setminus F} f(x^{(i)}) + \sum_{i \in F \cap [k]} (1-f(x^{(i)})) + (1-1) \geq 1$,
  which is equivalent to Inequality~\eqref{InequalityOrderedParity} for $F \cap [k]$
  in the projection. Note that $F \cap [k]$ has even cardinality.
\end{proof}

\pagebreak[3]
\section{Separation}
\label{SectionSeparation}

In this section we show how to solve the separation problem for Inequalities~\eqref{InequalityOrderedParity}.
\begin{theorem}
  Let $r \in \N^k$ and let $\hat{x} = (\hat{x}^{(1)}, \dotsc, \hat{x}^{(k)}) \in \Pord{r_1} \times \dotsc \times \Pord{r_k}$.
  We can decide in linear time if $\hat{x} \in \Peven{r}$ holds, and if this is not the case,
  obtain an odd set $F \subseteq [k]$ whose associated Inequality~\eqref{InequalityOrderedParity}
  is violated by $\hat{x}$.
\end{theorem}

\DeclareDocumentCommand\symDiff{}{\Delta}

\begin{proof}
  It is easy to see that it suffices to
  compute $\hat{\lambda}_i := f(\hat{x}^{(i)}) \in [0,1]$ for every $i \in [k]$ and then
  to minimize $\sum_{i \in [k] \setminus F} \hat{\lambda}_i + \sum_{i \in F} (1 - \hat{\lambda}_i)$
  over all $F \subseteq [k]$ of odd cardinality.
  This is done by computing the set $F' := \setdef{ i \in [k] }[ \hat{\lambda}_i > \frac{1}{2} ]$.
  If $|F'|$ is odd, then $F := F'$ is the minimizer.
  Otherwise, we let $F := F' \symDiff \{\hat{i}\}$ (by $\symDiff$ we denote the symmetric difference, i.e., $A \symDiff B := (A \cup B) \setminus (A \cap B)$)
  where $\hat{i} \in [k]$ is the index for which $|\hat{\lambda}_{\hat{i}} - \frac{1}{2}|$ is minimum.

  The computation of $\hat{\lambda}$ can obviously be done in linear time.
  Furthermore, the construction of $F'$ and the search for $\hat{i}$ can be carried out in time $\orderO{k}$
  which results in a linear total running time.
\end{proof}

Directly from Corollary~\ref{TheoremOrderedParityOuterDescriptionOdd} we obtain the separation algorithm for
the ordered \emph{odd} parity polytopes. The only difference to the even case
is that $|F|$ must be even, i.e., we set $F := F' \symDiff \{\hat{i}\}$ (for the same $\hat{i}$) if and only if $|F'|$ is odd.

\begin{corollary}
  Let $r \in \N^k$ and let $\hat{x} = (\hat{x}^{(1)}, \dotsc, \hat{x}^{(k)}) \in \Pord{r_1} \times \dotsb \times \Pord{r_k}$.
  We can decide in linear time if $\hat{x} \in \Podd{r}$ holds, and if this is not the case,
  obtain an even set $F \subseteq [k]$ whose associated Inequality~\eqref{InequalityOrderedParity}
  is violated by $\hat{x}$.
\end{corollary}

\section{Strengthened Blossom Inequalities for the Graphic TSP}
\label{SectionGTSP}

In this section we consider the \emph{graphic traveling salesman problem (GTSP)},
defined as follows.
The input consists of an undirected graph $G = (V,E)$, and the goal
is to find a minimum-length closed walk in $G$
that visits every node at least once.
Recently,~\citeauthor{SeboV14} developed a 7/5-approximation algorithm for this problem~\cite{SeboV14}.

When modeling it as an integer program (IP) one has different options.
The first is based on the observation that the problem is equivalent
to the traveling salesman problem on the complete graph with $|V|$ nodes
and edge weights $c\in \R_+^E$ where $c_{u,v}$ is equal to the
(combinatorial) distance from $u$ to $v$ in $G$.
This model has $\bincoeff{|V|}{2}$ binary variables which is much greater than $|E|$ if $G$ is sparse.

\DeclareDocumentCommand\GcutEdges{}{\delta}

In order to model the problem with fewer variables we can start with the IP
\begin{alignat}{10}
  &\text{min }  & \sum_{e \in E} z_e \label{LPGraphicTSPObjectiveFirst} \\
  &\text{s.t. } & z(\GcutEdges(S))    & \geq 2 && \qquad\text{for all } \emptyset \neq S \subsetneqq V \label{LPGraphicTSPConnectivity} \\
  &             & z_e                 & \geq 0 && \qquad\text{for all } e \in E \label{LPGraphicTSPBounds} \\
  &             & z_e                 & \in \Z && \qquad\text{for all } e \in E, \label{LPGraphicTSPIntegralityLast}
\end{alignat}
where by $\GcutEdges(\cdot)$ we denote the cut-sets.
The solutions to this problem correspond to the $2$-edge-connected spanning subgraphs of $G$.
Since not every such subgraph is a closed walk, this IP does not model the GTSP.
The necessary additional requirement is the parity condition $z(\GcutEdges(v)) \in 2\Z$ for every $v \in V$.
Such a constraint cannot be modeled directly, i.e., by adding linear inequalities.
To obtain a correct model, one may add integer variables $y_v = \frac{1}{2} z(\GcutEdges(v))$ for every $v \in V$.
Unfortunately, this does not contribute to the strength of the LP relaxation since relaxing $y_v$'s integrality constraints
essentially makes them redundant.

It is well-known that an optimum solution will never use an edge more than twice
(otherwise we can decrease the value by $2$ and obtain a better feasible solution).
Hence we can restrict $z_e$ to the set $\setdef{0,1,2}$ which allows us to
perform binarization, i.e., to replace, for every $e \in E$, the variable
$z_e$ by $x_e^{(1)} + x_e^{(2)}$ with $x_e^{(1)}, x_e^{(2)} \in \setdef{0,1}$.
This again leads to the same strength of the LP relaxation, but now allows us to
enforce parity constraints:
since every closed walk traverses each cut $\GcutEdges(S)$
an even number of times, the following inequalities are valid
for every $S \subseteq V$ and every $F \subseteq \GcutEdges(S) \times \setdef{0,1}$
with $|F|$ odd.
\begin{gather}
  \label{InequalityBlossomOriginal}
  \sum_{ (e,i) \in (\GcutEdges(S) \times \setdef{0,1}) \setminus F} x^{(i)}_e + \sum_{(e,i) \in F} (1-x^{(i)}_e) \geq 1
\end{gather}
If we identify the variables $x^{(1)}_e$ and $x^{(2)}_e$ with the doubled edges of $G$,
then these inequalities are the well-known \emph{Blossom Inequalities}.
These enforce that, for every cut $\delta(S)$, the projection of $x$ onto the variables corresponding to $\delta(S)$ is in the even parity polytope.

We will soon strengthen these inequalities. In order to argue later why the separation problem for the strengthened inequalities can be solved efficiently, we review the separation problem for the unstrengthened ones.
The first separation algorithm for Inequalities~\eqref{InequalityBlossomOriginal} was developed
by~\citeauthor{PadbergR82}~\cite{PadbergR82}.
The asymptotically fastest separation algorithm is by~\citeauthor{LetchfordRT08}~\cite{LetchfordRT08}.
More precisely, they describe how to efficiently solve (a generalization of) the following problem.
Given a graph $\overline{G} = (\overline{V},\overline{E})$ and weight vectors $c,c' \in \R_+^{\overline{E}}$, their algorithm
can minimize the objective
\begin{gather}
  \beta(S,F) := \sum_{e \in \GcutEdges(S) \setminus F} c_e + \sum_{e \in F} c'_e
\end{gather}
over all sets $S \subseteq \overline{V}$ and all odd sets $F \subseteq \GcutEdges(S)$.
Clearly, by letting $\overline{G}$ be the graph obtained from $G$ by doubling every edge and identifying the two edges
$e_1,e_2 \in \overline{E}$ with the variables $x^{(1)}_e$ and $x^{(2)}_e$ for $e \in E$, respectively,
we can solve the separation problem for a point $\hat{x} \in [0,1]^{2E}$ by minimizing $\beta(S,F)$
with the following weight vectors:
\begin{gather*}
  c_{e_i} := \hat{x}^{(i)}_e \quad\text{ and }\quad c'_{e_i} := 1 - \hat{x}^{(i)}_e \quad \text{ for all $i=1,2$ and all $e \in E$}
\end{gather*}
In fact we do not have to double the graph's edges for the computation: the cases where $F$ contains none or both edges contribute the same to the parity constraint on $F$.
Hence, we only have to consider the case that has lower weight.
Similarly, only one of the two cases of $F$ containing precisely one of the two edges needs to be considered.
This can be implemented using the weights
\begin{gather*}
  c_e := \min( \hat{x}^{(1)}_e + \hat{x}^{(2)}_e, 2-\hat{x}^{(1)}_e -\hat{x}^{(2)}_e ) \quad\text{ and } \quad c'_e := \min( 1-\hat{x}^{(1)}_e +\hat{x}^{(2)}_e , 1+\hat{x}^{(1)}_e -\hat{x}^{(2)}_e )
\end{gather*}
for all $e \in E$.

\bigskip

\paragraph{Strengthened Constraints.}
To break the symmetry of $x^{(1)}_e$ and $x^{(2)}_e$ we add $x^{(1)}_e \geq x^{(2)}_e$ for each edge $e \in E$.
Thus, we can now enforce that, for every cut $\delta(S)$, the projection of $x$ onto the variables corresponding to $\delta(S)$ is in the ordered even parity polytope $\Peven{r}$, where $r = (2,2,\dotsc,2) \in \Z^{|\delta(S)|}$.
By Theorem~\ref{TheoremOrderedParityOuterDescriptionEven}, the strengthened version of Constraint~\eqref{InequalityBlossomOriginal} reads
\begin{gather}
  \label{InequalityBlossomStrengthened}
  \sum_{ e \in \GcutEdges(S) \setminus F} (x^{(1)}_e - x^{(2)}_e) + \sum_{e \in F} (1-x^{(1)}_e + x^{(2)}_e) \geq 1
\end{gather}
for every $S \subseteq V$ and every $F \subseteq \GcutEdges(S)$ with $|F|$ odd.
Using almost the same separation procedure we can solve the separation problem
by working in the original graph $G$, and defining the weight vectors via
\begin{gather*}
  c_e := \hat{x}^{(1)}_e - \hat{x}^{(2)}_e \quad\text{ and }\quad c'_e := 1 - \hat{x}^{(1)}_e + \hat{x}^{(2)}_e \quad \text{ for all $e \in E$} \ .
\end{gather*}
Note that the computational costs for separating Constraints~\eqref{InequalityBlossomOriginal} and~\eqref{InequalityBlossomStrengthened}
are almost equal.

\pagebreak[3]
\section{Strength of the Relaxation}
\label{SectionCheating}

We carried out computational experiments on the GTSP, and it turned out that
the lower bound obtained by the linear relaxation \eqref{LPGraphicTSPObjectiveFirst}--\eqref{LPGraphicTSPBounds}
was never improved by binarization of the variables and addition of (strengthened) parity constraints.
In fact, in the root node many (strengthened) Blossom Inequalities were added, but with no effect on the dual objective.

We investigated this effect, and observed the following weakness of the approach of
binarization and addition of parity constraints. The decisive property of our model is that
\emph{only} parity constraints actually use the binarization variables $x_e^{(i)}$: all other
variables consider the original variables $z_e$ (or, equivalently, the sum of the binarization variables $x_e^{(i)}$).

In a more abstract setting we consider arbitrary (integer) variables $z_1, \dotsc, z_k$ with domains
$z_i \in [0,r_k] \cap \Z$ for all $i \in [k]$, and apply binarization, i.e.,
introduce variables $x_j^{(i)} \in \Pord{r_i}$ and linking constraints $z_i = \sum_{j=1}^{r_i} x_j^{(i)}$
for each $i \in [k]$.
We can assume that further (arbitrary) constraints that link the $z$-variables are also present.

Now consider a parity constraint on (a subset of) the $z$-variables, which is of course stated in terms of
the corresponding $x$-variables.
Suppose we have a certain fractional relaxation solution $(\hat{z},\hat{x})$ that may violate this parity constraint,
but satisfies all other constraints.
If we can modify the $x$-variables such that the linking constraints are still satisfied (i.e., that
the sums $\sum_{j=1}^{r_i} \hat{x}_j^{(i)}$ remain constant) and such that the parity constraint
is satisfied, then we obtain a feasible solution of the same value.
Unfortunately, such a modification is possible under very mild conditions.
In fact, the modifications are
very general in the sense that they can be done independently for every set of binarization variables,
that is, for all variables $x_j^{(i)}$ for a \emph{single} $i \in [k]$.
Furthermore, they do neither depend on the parity inequality that is actually violated nor on the parity.
Hence, even if there exist many parity constraints whose variable sets overlap, such modifications can still
exist, as it is often the case for Blossom Inequalities (note that there actually exist $2^{|V|-1}$ such parity constraints).
We will now describe the modifications, and later discuss why the conditions are often satisfied for solutions
of the linear relaxation \eqref{LPGraphicTSPObjectiveFirst}--\eqref{LPGraphicTSPBounds}.

\begin{lemma}
  \label{TheoremSingleParityCheat}
  Let $n \in \N$, and $z \in [0,n]$.
  Then the maximum value of $\min(f(x),1-f(x))$
  over all $x \in \Pord{n}$ with $z = \sum_{i=1}^n x_i$
  is equal to the minimum of $z$, $n-z$, and $\frac{1}{2}$.
\end{lemma}

\begin{proof}
  We denote the maximum in the lemma by $\gamma^*$, and by definition we have $\gamma^* \in [0,\frac{1}{2}]$.
  We distinguish three cases, depending on where the minimum of $z$, $n-z$, and $\frac{1}{2}$ is attained.

  \medskip
  \noindent
  \textbf{Case 1}: $z < \frac{1}{2}$.
  Clearly, since $f(x) = x_1 - x_2 + x_3 - \dotsb \pm x_n \leq x_1 + x_2 + x_3 + \dotsb + x_n = z$
  holds, we have $\gamma^* \leq z$.
  This value is obtained by $x = (z,0,\dotsc,0,0) \in \Pord{n}$ which satisfies $\sum_{i=1}^n x_i = z$ and $f(x) = z$.

  \medskip
  \noindent
  \textbf{Case 2}: $n-z < \frac{1}{2}$.
  This implies $x_i \geq \frac{1}{2}$ and hence $x_i \leq 2 - x_i$ for all $i \in [n]$.
  For even $n$ we obtain $f(x) \leq (2-x_1) - x_2 + (2-x_3) - x_4 + \dotsb + (2-x_{n-1}) - x_n \leq n - z$.
  For odd $n$, $1 - f(x) \leq 1 - ( x_1 - (2-x_2) + x_3 - (2-x_4) + \dotsb + x_{n-1} ) \leq n - z$.
  Together, this
  proves $\gamma^* \leq n-z$.
  This value is obtained by $x = (1,1,\dotsc,1,z-n+1) \in \Pord{n}$ which satisfies $\sum_{i=1}^n x_i = z$.
  Depending on the parity of $n$ it also satisfies
  $f(x) = z-n+1$ or $1 - f(x) = z-n+1$, i.e., we have $\gamma^* = n-z$.

  \medskip
  \noindent
  \textbf{Case 3}: $\min(z,n-z) \geq \frac{1}{2}$.
  Since $\gamma^* \leq \frac{1}{2}$ always holds, it remains to construct a solution of that value.
  We define $k \in \N$ to be an integer whose distance to $z$ is at most $\frac{1}{2}$,
  breaking ties such that $k$ is not equal to $0$ or equal to $n$ (but arbitrary otherwise).
  We furthermore define the values $z^-_k := \frac{2z - 2k + 1}{4} \in [0,\frac{1}{2}]$, and $z^+_k := \frac{2z - 2k + 3}{4} \in [\frac{1}{2}, 1]$.
  We again consider two cases, depending on $k$:

  \medskip
  \noindent
  \textbf{Case 3 (a)}: $1 \leq k \leq n-2$.
  We consider the vector $x$ defined via
  $$x_i := \begin{cases}
    1           & \text{ for } i \leq k-1 \\
    \frac{1}{2} & \text{ for } i = k \\
    z^-_k       & \text{ for } i \in \setdef{k+1,k+2} \\
    0           & \text{ for } i \geq k+3
  \end{cases},$$
  and observe that $x \in \Pord{n}$ holds, and that the first $(k-1)$ indices have $1$'s.
  Furthermore, we have that $\sum_{i=1}^{n} x_i = (k-1) + \frac{1}{2} + 2z^-_k = k-1 + \frac{1}{2} + \frac{ 2z - 2k + 1 }{ 2 } = z$
  and $f(x) = \frac{1}{2}$ hold.

  \medskip
  \noindent
  \textbf{Case 3 (b)}: $k = n-1$.
  We consider the vector $x = (1,1,\dotsc,1,z^+_{n-1},z^+_{n-1},\frac{1}{2})$
  and observe that $x \in \Pord{n}$ holds.
  Again, we have that $\sum_{i=1}^{n} x_i = (n-3) + 2z^+_{n-1} + \frac{1}{2} = n-3 + \frac{ 2z - 2(n-1) + 3 }{ 2 } + \frac{1}{2} = z$
  and $f(x) = \frac{1}{2}$ hold. This concludes the proof.
\end{proof}

Motivated by Lemma~\ref{TheoremSingleParityCheat} we define $\gamma(z) := \min(z,r-z,\frac{1}{2})$ for a variable $z \in [0,r]$.
In the following theorem we binarize a set of variables and consider parity constraints on subsets of them.

\begin{theorem}
  \label{TheoremMultiParityCheat}
  Let $r \in \N^k$ and let $z \in [0,r_1] \times \dotsb \times [0,r_k]$.
  Let $\mathcal{I}$ be a family of subsets $I \subseteq [k]$.
  If every $I \in \mathcal{I}$ satisfies $\sum_{i \in I} \gamma(z_i) \geq 1$,
  then there exist $x^{(i)} \in \Pord{r_i}$ with $z_i = \sum_{j=1}^{r_i} x^{(i)}_j$ for all $i \in [k]$
  such that for every $I \in \mathcal{I}$ the vector $(x^{(i)})_{i \in I}$
  is contained in the even and odd parity polytopes corresponding to $(r_i)_{i \in I}$.
\end{theorem}

\begin{proof}
  Let $x^{(i)}$ be the maximizer from Lemma~\ref{TheoremSingleParityCheat} for all $i \in [k]$,
  that is, $\min(f(x^{(i)}), 1-f(x^{(i)})) = \gamma(z_i)$ holds for every $i \in [k]$.
  Now consider one of the sets $I \in \mathcal{I}$ and any subset $F \subseteq I$.
  From
  \begin{gather*}
    1 \leq \sum_{i \in I} \gamma(z_i) = \sum_{i \in I} \min(f(x^{(i)}),1-f(x^{(i)})) \leq \sum_{i \in I \setminus F} f(x^{(i)}) + \sum_{i \in F} (1-f(x^{(i)}))
  \end{gather*}
  we obtain that Inequality~\eqref{InequalityOrderedParity} is satisfied, which concludes the proof.
\end{proof}

The theorem essentially states sufficient conditions for the case that after binarization
and enforcing of several parity constraints, the values of the original variables remain feasible.
Note that in order to satisfy $\sum_{i \in I} \gamma(z_i) \geq 1$ for a constraint on variable set $I$,
it suffices that two of the participating variables have a distance to their respective bounds of at least $\frac{1}{2}$ (see Lemma~\ref{TheoremSingleParityCheat}),
which is not very restrictive for nonbinary variables.

\bigskip

\paragraph{Implications for the Graphic TSP.}
We consider an optimum solution $\hat{z} \in [0,2]^E$ of the LP relaxation~\eqref{LPGraphicTSPObjectiveFirst}--\eqref{LPGraphicTSPBounds} 
of the model introduced in Section~\ref{SectionGTSP}.
The requirements for Theorem~\ref{TheoremMultiParityCheat}
are satisfied if and only if for every nontrivial cut $\GcutEdges(S)$ ($S \subsetneqq V$, $S \neq \emptyset$),
the inequality $\sum_{e \in \GcutEdges(S)} \gamma(\hat{z}_e) \geq 1$ is satisfied.
This is in particular the case if every cut contains two edges $e,f$ with $\frac{1}{2} \leq \hat{z}_e,\hat{z}_f \leq \frac{3}{2}$.
Note that integrality of the $z$-variables does not play a role here, i.e., even if $\hat{z}$ is integral, and $\hat{z}(\GcutEdges(S))$ has the wrong
parity for some set $S$, then there may exist a (fractional) assignment for $x$-variables that is feasible for
Constraints~\eqref{InequalityBlossomStrengthened}.
For the IP model this means that the Blossom Inequalities only become useful if relevant $z$-variables are near their bounds
or if branching or cutting restricted the $x$-variables further.

\bigskip

\paragraph{Acknowledgements.}
We thank the author of the \texttt{lrs} software~\cite{Avis00} that we used to
compute complete descriptions for ordered parity polytopes of small dimensions.
Furthermore, we thank Stefan Weltge and Volker Kaibel for valuable discussions.
Finally, we are very grateful to the two referees who found a mistake in the proof of Lemma~\ref{TheoremSingleParityCheat} and whose comments led to improvements in the presentation of the material.

\bibliographystyle{plainnat}
\bibliography{references}

\end{document}